\documentclass[letterpaper]{amsart}

\usepackage{graphics}
\usepackage{amssymb}

\newtheorem{lemma}{Lemma}
\newtheorem{theorem}{Theorem}
\newtheorem{corollary}{Corollary}
\newtheorem{proposition}{Proposition}
\newtheorem{problem}{Problem}
\theoremstyle{definition}
\newtheorem{definition}{Definition}
\newtheorem{example}{Example}

\DeclareMathOperator*\Fin{Fin}
\DeclareMathOperator*\ran{ran}
\renewcommand\and{\text{ and }}
\begin{document}
\title{Coexistence in interval effect algebras}
\author{Gejza Jen\v ca}
\address{
Department of Mathematics and Descriptive Geometry\\
Faculty of Civil Engineering\\
Radlinsk\' eho 11\\
	Bratislava 813 68\\
	Slovak Republic
}
\email{gejza.jenca@stuba.sk}
\thanks{
This research is supported by grants VEGA G-1/3025/06,G-1/0500/09 of M\v S SR,
Slovakia and by the Slovak Research and Development Agency under the contracts
No. APVT-51-032002, APVV-0071-06.
}
\begin{abstract}
Motivated by the notion of coexistence of effect-valued observables,
we give a characterization of coexistent subsets of interval effect algebras.
\end{abstract}
\maketitle

\section{Introduction and motivation}

Let $\mathbb H$ be a Hilbert space, let $\mathcal S(\mathbb H)$ be the
partially ordered abelian group of all bounded self-adjoint operators on
$\mathbb H$ equipped with the usual order.  We write $\mathcal E(\mathbb H)$
for the closed interval $[0,I]_{\mathcal S(\mathbb H)}=\{A\in\mathcal S(\mathbb
H):0\leq A\leq I\}$, where $I$ is the identity operator on $\mathbb H$. 
The elements of $\mathcal E(\mathbb H)$ are called {\em Hilbert space effects}.

From the algebraic point of view, $\mathcal E(\mathbb H)$, equipped with addition,
is the canonical example of an {\em interval effect algebra} \cite{FouBen:EAaUQL},\cite{BenFou:IaSEA}.

Let $(\Omega,\mathcal A)$ be a measurable space.  An {\em observable} is a
mapping $\alpha:\mathcal A\to\mathcal E(\mathbb H)$ 
such that $\alpha(\emptyset)=0$, $\alpha(\Omega)=1$ and for every
pairwise disjoint system $\{X_i\}_{i\in\mathbb N}$ of measurable sets,
$\alpha(\bigcup_{i\in\mathbb N}X_i)=\sum_{i\in\mathbb N}\alpha(X_i)$.
That means, an observable is a normalized positive operator measure
$\alpha:\mathcal A:\to \mathcal E(\mathbb H)$.

In the unsharp observable approach to quantum mechanics, 
observables represent measurable quantities,
see for example \cite{BusLahMit:TQToM}, \cite{BusGraLah:OQP}.

\begin{definition}
Let $(\Omega_1,\mathcal A_1)$ and $(\Omega_2,\mathcal A_2)$ be measurable spaces.
Observables $\alpha_1:\mathcal A_1\to \mathcal E(\mathbb H)$ and
$\alpha_2:\mathcal A_2\to \mathcal E(\mathbb H)$ are {\em coexistent} if there is
a measurable space $(\Omega,\mathcal A)$ and an observable 
$\alpha:\mathcal A\to \mathcal E(\mathbb H)$ such that
$$
\ran(\alpha_1)\cup\ran(\alpha_2)\subseteq\ran(\alpha),
$$
where $\ran(.)$ denotes the range of an observable.
\end{definition}

In mathematical physics, the notion of coexistence of observables 
is used to describe 
the possibility of measuring together two quantities. The concept
of coexistence of observables is due to Ludwig \cite{Lud:VeaGdQuapT} and it was
further investigated in \cite{Hel:Ceiqm}, \cite{Kra:SEaO}, \cite{Lud:FoQM}.
See also \cite{LahPul:COaEiQM}, \cite{LahPulYli:COaEiCA}, \cite{LahPul:CvFCoQO},
\cite{Pul:CaDoE} for recent results on coexistence of observables
and \cite{Gud:CoQE}, \cite{BusSch:CoQE}, \cite{StaReiHeiZim:CoQE}
for investigations concerning coexistence of qubit effects.

Obviously, coexistence of any pair of observables depends only on their
ranges; two observables are coexistent if and only if the union of their ranges can be
embedded into a range of an observable. This fact gives rise to the following
problem.

\begin{problem}
Give a characterization of such sets $S$ 
of effects on a Hilbert space $\mathbb H$ such that
there is a measurable space $(\Omega,\mathcal A)$ and an observable 
$\alpha:\mathcal A\to\mathcal E(\mathbb H)$ with
$S\subseteq\ran(\alpha)$.
\end{problem}

If $S$ consists only of orthogonal projections (that means, idempotent effects),
then the answer is simple: $S$ is
a subset of the range of an observable if and only if the elements of $S$ commute.
On the other hand, if there are non-idempotent effects in $S$, the situation
is much more complicated. 

In the present paper, we examine a more general problem: 
\begin{problem}
Let $E$ be an interval effect algebra.
Characterize subsets $S$ of $E$ 
such that there is a Boolean algebra $B$ and a morphism of effect algebras 
$\alpha:B\to E$ with
$S\subseteq\ran(\alpha)$.
\end{problem}

This can be considered as an abstract algebraic version of Problem 1.  In the
present paper, we prove that, given a subset $S$ of an interval effect algebra
$E$, there exist a Boolean algebra $B$ and a morphism $\alpha:B\to E$ with
$S\subseteq\ran(\alpha)$ if and only if there is a mapping $\beta:\Fin(S)\to
E$ satisfying certain properties.  We call such mappings {\em witness
mappings}.
In \cite{Jen:CSMiEA}, we introduced and studied a similar but 
more complicated notion, called {\em compatibility support mappings}, 
in the context of general effect algebras.
In the present paper, the more restrictive setting of interval effect algebras
allows us to introduce a simpler and much more intuitive notion of a witness
map, and we can still apply our results to our most important example, namely
$\mathcal E(\mathbb H)$.

\section{Definitions and basic relationships}

\subsection{Effect algebras}

An {\em effect algebra} 
is a partial algebra $(E;\oplus,0,1)$ with a binary 
partial operation $\oplus$ and two nullary operations $0,1$ satisfying
the following conditions.
\begin{enumerate}
\item[(E1)]If $a\oplus b$ is defined, then $b\oplus a$ is defined and
		$a\oplus b=b\oplus a$.
\item[(E2)]If $a\oplus b$ and $(a\oplus b)\oplus c$ are defined, then
		$b\oplus c$ and $a\oplus(b\oplus c)$ are defined and
		$(a\oplus b)\oplus c=a\oplus(b\oplus c)$.
\item[(E3)]For every $a\in E$ there is a unique $a'\in E$ such that
		$a\oplus a'$ exists and $a\oplus a'=1$.
\item[(E4)]If $a\oplus 1$ is defined, then $a=0$.
\end{enumerate}

Effect algebras were introduced by Foulis and Bennett in their paper 
\cite{FouBen:EAaUQL}.

In their paper~\cite{KopCho:DP}, Chovanec and K\^ opka introduced
an essentially equivalent structure called {\em D-poset}. Their definition
is an abstract algebraic version the {\em D-poset of fuzzy sets},
introduced by K\^ opka in the paper~\cite{Kop:DPFS}.

Another equivalent structure was introduced by Giuntini and
Greuling in~\cite{GiuGre:TaFLfUP}. We refer to~\cite{DvuPul:NTiQS} for more
information on effect algebras and related topics.

The class of effect algebras is (essentially) 
a common superclass of several important classes of algebras: orthomodular lattices
\cite{Kal:OL,Ber:OLaAA}, orthoalgebras \cite{FouRan:OS1BC}, MV-algebras \cite{Cha:AAoMVL,Mun:IoAFCSAiLSC}. 
In particular, every Boolean algebra is an effect algebra, if we introduce a
partial operation $\oplus$ such that $a\oplus b$ exists if and only if $a\wedge b=0$ and then
we put $a\oplus b:=a\vee b$.

\subsection{Properties of effect algebras}

In an effect algebra $E$, we write $a\leq b$ if and only if there is
$c\in E$ such that $a\oplus c=b$.  It is easy to check that for every effect
algebra $E$, $\leq$ is a partial order on $E$.  Moreover, it is possible to introduce
a new partial operation $\ominus$; $b\ominus a$ is defined if and only if $a\leq
b$ and then $a\oplus(b\ominus a)=b$.  It can be proved that, in an effect
algebra, $a\oplus b$ is defined if and only if $a\leq b'$ if and only if $b\leq
a'$. In an effect algebra, we write $a\perp b$ if and only if $a\oplus b$ exists.

A finite family $(a_1,\dots,a_n)$ of elements of an effect algebra is called
{\em orthogonal} if and only if the sum $a_1\oplus\dots\oplus a_n$ exists. An orthogonal
family $(a_1,\dots,a_n)$ is a {\em decomposition of unit} if and only if $a_1\oplus\dots a_n=1$.

\subsection{Morphisms of effect algebras}

Let $E,F$ be affect algebras. A mapping $\phi:E\to F$ is a {\em morphism of effect algebras}
if and only if the following conditions are satisfied:
\begin{enumerate}
\item[(EM1)]$\phi(1)=1$.
\item[(EM2)]If $a,b\in E$, $a\perp b$ then $\phi(a)\perp\phi(b)$ and $\phi(a\oplus b)=\phi(a)\oplus\phi(b)$.
\end{enumerate}
We note that every morphism of effect algebras is isotone. Moreover, every 
morphism of effect algebras preserves the $0$ element, as well as the unary operation $x\mapsto x'$ and the partial binary operation $\ominus$.

A bijective morphism of effect algebras $\phi:E\to F$ such that $\phi^{-1}$ is a
morphism of effect algebras is called {\em an isomorphism of effect algebras}.

Let $B$ be a Boolean algebra and let $E$ be an effect algebra. A morphism of
effect algebras $\alpha:B\to E$ is called {\em an observable}. If $B$ is
finite, then we say that $\alpha$ is a {\em a simple observable}.

\begin{definition}
We say that a subset $S$ of an effect algebra is {\em coexistent} if there exists
a Boolean algebra $B$ and an observable $\alpha:B\to E$ such that $S\subseteq\alpha(B)$.
\end{definition}

If is easy to check that a subset $\{a,b\}$ of an effect algebra $E$ is coexistent if and only if
there is $c\in E$ such that $c\leq a,b$ and $a\perp b\ominus c$. Following 
the terminology in \cite{Gud:CoQE},
we say that $c$ is {\em the witness element for $a,b$.}

\subsection{Partially ordered abelian groups}

Let $G$ be an (additive) abelian group. We say that $G$ is a {\em partially ordered abelian group} if and only if $G$ is equipped with a partial order that is compatible with addition,
that means, for all $a,b,t\in G$,
$$
a\geq b\Longrightarrow a+t\geq b+t.
$$

\begin{example}\label{ex:selfadjoint}
Let $\mathbb H$ be a Hilbert space, let $\mathcal S(\mathbb H)$ be the
set of all bounded self-adjoint operators on $\mathbb H$. For 
$A,B\in\mathcal S(\mathbb H)$, write $A\leq B$ if and only if,
for all $x\in\mathbb H$, $\langle Ax,x\rangle\leq\langle Bx,x\rangle$.
Then $(\mathcal S(\mathbb H),+,0)$ is a partially ordered abelian group.
\end{example}

For a partially ordered abelian group $G$, we write 
$$
G^+=\{a\in G: a\geq 0\}.
$$
The elements of $G^+$ are called {\em positive}.
Obviously, $G^+$ is a submonoid of $G$. Moreover, $G^+$ is {\em conical}, that
means, if $a,b\in G^+$ and $a+b=0$, then $a=b=0$. 

It is easy to see that
there is a one-to one correspondence between partial orders on $G$ and conical 
submonoids of $G$.

\subsection{Order units}

Let $G$ be a partially ordered abelian group.
We say that $u\in G^+$ is an order unit if and only if for every $a\in G$ there is
$n\in\mathbb N$ such that $n.u\geq a$. 

A pair $(G,u)$, where $G$ is a partially ordered abelian group and $u$ 
is an order unit of $G$ is called {\em a unital group}.

Let $(G_1,u_1)$, $(G_2,u_2)$ be unital groups. A mapping $\phi:G_1\to G_2$ is
a {\em morphism of unital groups} if and only if  $\phi$ is a group homomorphism,
$x\geq y$ implies $\phi(x)\geq\phi(y)$ and $\phi(u_1)=u_2$.

For a morphism of unital groups, we write $\phi:(G_1,u_1)\to(G_2,u_2)$.

\begin{example}
In $\mathcal S(\mathbb H)$, the identity operator is
an order unit.
\end{example}

\subsection{Interval effect algebras}

One can construct examples of effect algebras from an arbitrary partially
ordered abelian group $(G,\leq)$ in the following way: Choose any positive
$u\in G$; then, for $0\leq a,b\leq u$, define $a\oplus b$ if and only if
$a+b\leq u$ and put $a\oplus b=a+b$.  With such partial operation $\oplus$, the
closed interval 
$$
[0,u]_G=\{x\in G:0\leq x\leq u\}
$$ 
becomes an effect algebra $([0,u]_G,\oplus,0,u)$.  Effect
algebras which arise from partially ordered abelian groups in this way are
called {\em interval effect algebras}, see \cite{BenFou:IaSEA}.

More generally, we say that {\em $E$ is an interval effect algebra in the
group $G$} if $E$ is of the form $[0,u]_{G^+}$, of isomorphic to such 
an effect algebra.

\begin{example}
The prototype interval effect algebra is 
the {\em standard effect algebra}
$\mathcal E(\mathbb H)=[0,I]_{\mathcal S(\mathbb H)}$.
\end{example}

\subsection{MV-algebras and MV-effect algebras}

An {\em MV-algebra} (c.f. \cite{Cha:AAoMVL}, \cite{Mun:IoAFCSAiLSC}) is a
$(2,1,0)$-type algebra $(M;\boxplus,\lnot,0)$, 
such that $\boxplus$
satisfies the
identities $(x\boxplus y)\boxplus z=x\boxplus (y\boxplus z)$,
$x\boxplus y=y\boxplus x$, $x\boxplus0=x$, $\lnot\lnot x=x$, $x\boxplus\lnot 0=\lnot 0$ and
$$
x\boxplus\lnot(x\boxplus\lnot y)=y\boxplus\lnot(y\boxplus\lnot x)\text{.}
$$
On every MV-algebra, a partial order $\leq$ is defined
by the rule
$$
x\leq y\Longleftrightarrow y=x\boxplus\lnot(x\boxplus\lnot y).
$$
In this partial order, every MV-algebra is a distributive lattice bounded by
$0$ and $\lnot 0$.

An {\em MV-effect algebra} is a lattice ordered effect algebra $M$ in which,
for all $a,b\in M$, 
\begin{equation}\label{eq:MV}
(a\lor b)\ominus a=b\ominus (a\land b).
\end{equation}
 It is proved
in \cite{ChoKop:BDP} that there is a natural, one-to one correspondence between
MV-effect algebras and MV-algebras given by the following rules.
Let $(M,\oplus,0,1)$ be an MV-effect algebra. Let $\boxplus$ 
be a total operation given by $x \boxplus y=x\oplus(x'\land y)$. Then
$(M,\boxplus,',0)$ is an MV-algebra. Similarly, let 
$(M,\boxplus,\lnot,0)$ be an MV-algebra. Restrict the operation
$\boxplus$ to the pairs $(x,y)$ satisfying $x\leq\lnot y$ and call the
new partial operation $\oplus$. Then $(M,\oplus,0,\lnot 0)$ is an MV-effect algebra.

We note that every Boolean algebra is an MV-effect algebra. 

According to \cite{Mun:IoAFCSAiLSC}, every MV-effect algebra is isomorphic to an interval $[0,u]_G$, where $G$ is a lattice-ordered group. Thus, every MV-effect algebra is an interval algebra.

\subsection{Group valued measures and ambient groups}

Let $E$ be an effect algebra and let $(G_2,u_2)$ be a unital group. A morphism
of effect algebras from $E$ to the interval effect algebra $[0,u_2]_{G_2}$ 
is called a {\em group-valued measure}.

\begin{proposition}\label{prop:ambient}
\cite{BenFou:IaSEA}
Let $E$ be an interval effect algebra. There exists a unital group
$(G_1,u_1)$ such that $E=[0,u_1]_{G_1}$,
$E$ generates $G_1$ and for every unital group $(G_2,u_2)$ and
every group valued measure $\phi:E\to [0,u_2]_{G_2}$, $\phi$ 
extends to a unique morphism of unital groups $\widehat\phi:(G_1,u_1)\to (G_2,u_2)$.
The unital group $(G_1,u_1)$ is unique, up to isomorphism.
\end{proposition}

The unital group $(G_1,u_1)$ from Proposition~\ref{prop:ambient} is
called {\em the ambient group} of $E$, denoted by $G(E)$.

\begin{example}
$\mathcal S(\mathbb H)$ is the ambient group of $\mathcal E(\mathbb H)$.
\end{example}

\subsection{M\"obius inversion theorem}

We say that a partially ordered set $(P,\leq)$ is {\em locally finite} if and only if
every closed interval 
$$
[x,y]_P:=\{z\in P:x\leq z\leq y\}
$$
is a finite set.

let $(P,\leq)$ be a locally finite 
partially ordered set. Define $I(P)$ to be the set of all pairs
$(x,y)\in P\times P$ such that $(x\leq y)$.

There exists a unique function $\mu:I(P)\to\mathbb Z$ such that, for all $(x,y)\in I(P)$,
\begin{equation}\label{eq:mobius}
\sum_{x\leq z \leq y}\mu(x,z)=\delta_{x,y},
\end{equation}

where $\delta_{x,y}$ is the Kronecker delta. We say the $\mu$ is the 
{\em M\"obius mapping of the poset $P$}. We refer to the classical paper
\cite{Rot:OtfoctIToMF} and to the monograph \cite{Sta:EC} 
for more information on M\"obius mappings and related topics.

\begin{example}
Let $S$ be a set, write $\Fin(S)$ for the set of all
finite subsets of $S$. For the poset $(\Fin(S),\subseteq)$,
we have $\mu(X,Z)=(-1)^{|X|+|Z|}$.
\end{example}

\begin{theorem}[M\"obius inversion formula]
Let $G$ be an abelian group and let $f:I(P)\to G$. 
Define $$f^{\leq}(x,y):=\sum_{x\leq z\leq y}f(z,y).$$
Then
$$
f(x,y)=\sum_{x\leq z\leq y}\mu(x,z)f^{\leq}(z,y).
$$
\end{theorem}

We say that $f(x,y)$ is the M\"obius inversion of $f^{\leq}$.

\section{Witness mappings}

Let $E$ be an interval effect algebra in a partially ordered abelian group $G$.
Let $S\subseteq E$. Let us write $\Fin(S)$ for the set of all
finite subsets of $S$. Obviously, $(\Fin(S),\subseteq)$ is a locally
finite poset.

For every mapping $\beta:\Fin(S)\to G$, we define a
mapping $D_\beta:I(\Fin(S))\to G$.
For $(X,A)\in I(\Fin(S))$,
the value
$D_\beta(X,A)\in G$ is given by the rule
$$
D_\beta(X,A):=\sum_{X\subseteq Z\subseteq A}(-1)^{|X|+|Z|}\beta(Z).
$$

Note that there is an obvious connection to M\"obius inversions:
define
$\hat\beta:I(\Fin(S))\to G$ by
$$
\hat\beta(X,A)=\beta(X).
$$
Then $D_\beta$ is the M\"obius inversion
of $\hat\beta$ with respect to the poset $(\Fin(S),\subseteq)$.
By the M\"obius inversion formula we see that
\begin{equation}\label{eq:decomp}
\beta(X)=\hat\beta(X,A)=\sum_{X\leq Z\leq A}D_\beta(Z,A),
\end{equation}
for any $A\supseteq X$.
In particular, $A:=X$ yields $\beta(X)=D_\beta(X,X)$.

\begin{lemma}
\label{lemma:formal}
Let $E$ be an interval effect algebra in a partially ordered abelian group $G$.
Let $S$ be a subset of $E$, let $\beta:\Fin(S)\to G$.
For all $c\in S\setminus A$,
$$
D_\beta(X,A)=D_\beta(X,A\cup\{c\})+ D_\beta(X\cup\{c\},A\cup\{c\}).
$$
\end{lemma}
\begin{proof}
The proof is purely formal.
Let us rewrite
$$
D_\beta(X,A\cup\{c\})=\sum_{X\subseteq Z\subseteq A\cup\{c\}}(-1)^{|X|+|Z|}\beta(Z).
$$
For any $Z$ in the sum, either $c\in Z$ or $c\notin Z$.
If $c\in Z$, then $X\cup\{c\}\subseteq Z \subseteq A\cup \{c\}$.
If $c\notin Z$, then $X\subseteq Z\subseteq A$.
Consequently,
\begin{align*}
D_\beta(X,A\cup\{c\})=&\sum_{X\subseteq Z\subseteq A}(-1)^{|X|+|Z|}\beta(Z)+
\sum_{X\cup\{c\}\subseteq Z\subseteq A\cup \{c\}}(-1)^{|X|+|Z|}\beta(Z)=\\
=&D_\beta(X,A)+\sum_{X\cup\{c\}\subseteq Z\subseteq A\cup \{c\}}(-1)^{|X|+|Z|}\beta(Z)
\end{align*}
It remains to observe that
\begin{align*}
~&\sum_{X\cup\{c\}\subseteq Z\subseteq A\cup \{c\}}(-1)^{|X|+|Z|}\beta(Z)=\\
=&-\sum_{X\cup\{c\}\subseteq Z\subseteq A\cup \{c\}}(-1)^{|X\cup\{c\}|+|Z|}\beta(Z)=
D_\beta(X\cup\{c\},A\cup\{c\}).
\end{align*}
\end{proof}

\begin{definition}\label{def:cm}
Let $E$ be an interval effect algebra, let $S\subseteq E$.

We say that a mapping $\beta:\Fin(S)\to E$ is a {\em witness
mapping for $S$} if and only if the following conditions are satisfied.
\begin{enumerate}
\item[(A1)]$\beta(\emptyset)=1$,
\item[(A2)]for all $c\in S$, $\beta(\{c\})=c$,
\item[(A3)]for all $(X,A)\in I(\Fin(S))$, $D_\beta(X,A)\geq 0$.
\end{enumerate}
\end{definition}

Let us prove that our notion of a witness mapping can be considered as an extension
of the notion of a witness element.

\begin{proposition}
Let $E$ be an interval effect algebra in a partially ordered abelian group $G$.
Let $\{a,b\}$ be a subset of $E$ and let $\beta:\Fin(\{a,b\})\to E$
be a mapping satisfying the conditions (A1) and (A2) of Definition \ref{def:cm}.
Then $\beta$ is a witness map for $\{a,b\}$ if and only if $\beta(\{a,b\})$ 
is a witness for $a,b$.
\end{proposition}
\begin{proof}
The mapping $\beta:\Fin(\{a,b\})\to E$ given by
$$
\begin{array}{c|c|c|c|c}
X & \emptyset & \{a\} & \{b\} & \{a,b\}\\
\hline
\beta(X) & 1 & a & b & c 
\end{array}
$$

Suppose that $c$ is a witness element.
Let us prove the condition (A3).
If $X=A$, there is nothing to prove. If $X=\emptyset$ and $A=\{x\}$ then
$D_\beta(X,A)=1-x\geq 0$ since $x\leq 1$. If $X=\{x\}$ and $A=\{a,b\}$ then
$D_\beta(X,A)=x-c\geq 0$, since $c\leq x$. If $X=\emptyset$ and $A=\{a,b\}$
then
$$
D_\beta(X,A)=D_\beta(\emptyset,\{a,b\})=
1-a-b+c.
$$
As $c$ is a witness element for $a,b$, $a\ominus c\perp b$. This can be written
as $a\ominus c\leq 1\ominus b$. Therefore, $(1\ominus b)\ominus (a\ominus c)\geq 0$ and
we may compute in $G$
\begin{equation}\label{eq:refme}
(1\ominus b)\ominus(a\ominus c)=1-b-a+c=1-a-b+c=D_\beta(\emptyset,\{a,b\}).
\end{equation}
Thus, $D_\beta(\emptyset,\{a,b\})\geq 0$.

Suppose that $\beta$ is a witness map.
As $D_\beta(\{a\},\{a,b\})=a-c\geq 0$, we see that $a\geq c$. Similarly, $b\geq c$.
As $D_\beta(\emptyset,\{a,b\})=1-a-b+c\geq 0$,
the equality (\ref{eq:refme}) implies that $a\ominus c\perp b$. Thus, $c$ is a witness
element.
\end{proof}
As an obvious consequence, we obtain the following
\begin{corollary}
Let $\{a,b\}$ be a subset of an interval effect algebra $E$.
Then $\{a,b\}$ is coexistent if and only if there is a witness mapping for $\{a,b\}$.
\end{corollary}

\subsection{Properties of witness mappings}

To shorten our formulations, let us introduce some running notation:
\begin{itemize}
\item $E$ is an interval effect algebra,
\item $S$ is a subset of $E$, 
\item $\beta:\Fin(S)\to E$ is a witness mapping for $S$.
\end{itemize}

Let us prove that $D_\beta$ is, in fact, an $E$-valued mapping.

\begin{proposition}
\label{prop:Diseffect}
For all $(X,A)\in I(\Fin(S))$, $D_\beta(X,A)\leq 1$.
\end{proposition}
\begin{proof}
By equation (\ref{eq:decomp}),
$$
\beta(X)=\hat\beta(X,A)=\sum_{X\leq Z\leq A}D_\beta(X,Z).
$$
In particular, $D_\beta(X,A)\leq\beta(X)\leq 1$.
\end{proof}

Since $D_\beta$ is an $E$-valued mapping, the $+$ operation
mentioned in Lemma~\ref{lemma:formal} is an $\oplus$ partial operation of the
effect algebra $E$:
\begin{proposition}
\label{prop:base}
For all $(X,A)\in I(\Fin(S))$ and $c\in S\setminus A$,
$D_\beta(X,A\cup\{c\})\perp D_\beta(X\cup\{c\},A\cup\{c\})$ and
$$
D_\beta(X,A)=D_\beta(X,A\cup\{c\})\oplus D_\beta(X\cup\{c\},A\cup\{c\}).
$$
\end{proposition}
\begin{proof}
By Lemma~\ref{lemma:formal},
$$
D_\beta(X,A)=D_\beta(X,A\cup\{c\})+D_\beta(X\cup\{c\},A\cup\{c\}).
$$
By Definition~\ref{def:cm} and Proposition~\ref{prop:Diseffect},
$D_\beta:I(Fin(S))\to E$.
\end{proof}
\begin{lemma}\label{lemma:zero}
Let $X\in\Fin(S)$. If $1\in S\setminus X$, then $D_\beta(X,X\cup\{1\})=0$.
\end{lemma}
\begin{proof}
(By induction with respect to $|X|$.)
If $X=\emptyset$, then 
$$
D_\beta(X,X\cup\{1\})=D_\beta(\emptyset,\{1\})=\beta(\emptyset)-\beta(\{1\})=1-1=0.
$$

Suppose that the Lemma is true for some $X$ and let $c\notin X$,
$c\neq 1$. We want to prove that $D_\beta(X\cup\{c\},X\cup\{c\}\cup\{1\})=0$.
Put $A=X\cup\{1\}$ in Proposition~\ref{prop:base} to obtain
$$
D_\beta(X,X\cup\{1\})=D_\beta(X,X\cup\{c\}\cup\{1\})\oplus D_\beta(X\cup\{c\},X\cup\{c\}\cup\{1\}).
$$
By the induction hypothesis, $D_\beta(X,X\cup\{1\})=0$, and since
$$
D_\beta(X,X\cup\{1\})\geq D_\beta(X\cup\{c\},X\cup\{c\}\cup\{1\}),
$$
we may conclude that $D_\beta(X\cup\{c\},X\cup\{c\}\cup\{1\})=0$.
\end{proof}

Later in Proposition \ref{prop:Dwedge} we will show that
for every MV-effect algebra $M$, $\bigwedge:\Fin(M)\to M$
is a witness mapping. The following proposition shows that
several properties of $\bigwedge$ are preserved for all witness
mappings.

\begin{proposition}~
\begin{enumerate}
\item[(a)]
$\beta$ is an antitone mapping from $(\Fin(S),\subseteq)$ to $(E,\leq)$.
\item[(b)]
For all $X\in\Fin(S)$, $\beta(X)$ is a lower bound of $X$.
\item[(c)]
Suppose that $0\in S$. If $0\in X\in\Fin(S)$, then $\beta(X)=0$.
\item[(d)]
Suppose that $1\in S$.
For all $X\in\Fin(S)$, $\beta(X)=\beta(X\cup\{1\})$
\end{enumerate}
\end{proposition}
\begin{proof}~
\begin{enumerate}
\item[(a)]
Let $X\in\Fin(S)$.
Let us prove that for any $c\in S\setminus X$, 
$\beta(X\cup \{c\})\leq\beta(X)$.
Put $X=A$ in Proposition~\ref{prop:base} to obtain
\begin{align*}
\beta(X)=D_\beta(X,X)=D_\beta(X,X\cup\{c\})\oplus 
	D_\beta(X\cup\{c\},X\cup\{c\})\geq\\
\geq D_\beta(X\cup\{c\},X\cup\{c\})
=\beta(X\cup\{c\}).
\end{align*}
The rest of the proof is a trivial induction.
\item[(b)]
Let $c\in X$. By (a), $\{c\}\subseteq X$ implies that
$$
c=\beta(\{c\})\geq\beta(X).
$$
\item[(c)]
Trivial, by (b).
\item[(d)]
If $1\in X$, there is nothing to prove.

Suppose that $1\not\in X$.
Putting $A=X$ and $c=1$ in Proposition~\ref{prop:base} yields
$$
D_\beta(X,X)=D_\beta(X,X\cup\{1\})\oplus D_\beta(X\cup\{1\},X\cup\{1\}).
$$
By Lemma~\ref{lemma:zero}, $D_\beta(X,X\cup\{1\})=0$, hence
$$
\beta(X)=D_\beta(X,X)=D_\beta(X\cup\{1\},X\cup\{1\})=\beta(X\cup\{1\}).
$$
\end{enumerate}
\end{proof}
\begin{proposition}\label{prop:pushbeta}
Let $E_1,E_2$ be interval effect algebras.
Let $\phi:E_1\to E_2$ be a
morphism of effect algebras. 
If $S_1\subseteq E_1$ is such that there is
a witness mapping $\beta_1$ of $S_1$, then $\phi(S_1)$ admits a
witness mapping.
\end{proposition}
\begin{proof}
The mapping $\phi$ is a $G(E_2)$ valued measure on $E_1$. Therefore,
there is a morphism of unigroups $\widehat\phi:(G(E_1),1)\to(G(E_2),1))$
extending $\phi$.

For every $a\in\widehat\phi(S_1)$, fix $p(a)\in S_1$ such that $\widehat\phi(p(a))=a$.
Define $\beta_2:\Fin(\widehat\phi(S_1))\to E_2$ as follows:
$$
\beta_2(\{x_1,\dots,x_n\})=\widehat\phi\bigl(\beta_1(\{p(x_1),\dots,p(x_n)\})\bigr),
$$
or, in other words, for $X\in\Fin(S_2)$,
$\beta_2(X)=\widehat\phi\bigl(\beta_1\bigl(p(X)\bigr)\bigr)$.  Then $\beta_2$
is a witness mapping for $\widehat\phi(S_1)$.

Indeed, the conditions (A1) and (A2) are trivially satisfied. 
For the proof of (A3) we may compute
\begin{align*}
D_{\beta_2}(X,A)=\sum_{X\subseteq Z\subseteq A}(-1)^{|X|+|Z|}\beta_2(Z)=
\sum_{X\subseteq Z\subseteq A}(-1)^{|X|+|Z|}
	\widehat\phi\bigl(\beta_1\bigl(p(Z)\bigr)\bigr)=\\
=\widehat\phi\bigl(\sum_{X\subseteq Z\subseteq A}(-1)^{|X|+|Z|}\beta_1
	\bigl(p(Z)\bigr)\bigr)=
\widehat\phi\bigl(\sum_{p(X)\subseteq Y\subseteq p(A)}(-1)^{|p(X)|+|Y|}\beta_1(Y)\bigr).
\end{align*}
Since $\beta_1$ is a witness mapping,
$(-1)^{|p(X)|+|Y|}\beta_1(Y)\in E_1$. 
Therefore,
$$
\widehat\phi\bigl(\sum_{p(X)\subseteq Y\subseteq p(A)}(-1)^{|p(X)|+|Y|}\beta_1(Y)\bigr)\in E_2.
$$
\end{proof}

\begin{proposition}\label{prop:restriction}
For every $S_0\subseteq S$, the restriction of $\beta$ to $\Fin(S_0)$ is
a witness mapping for $S_0$.
\end{proposition}
\begin{proof}
Trivial.
\end{proof}

\section{Examples of witness mappings}
To show that the notion of a witness mapping is natural, we present
two examples of witness mappings, arising from the meet operation on a MV-effect algebra
and the product operation on a commuting subset of $\mathcal E(\mathbb H)$.
\begin{proposition}
\label{prop:Dwedge}
Let $M$ be an MV-effect algebra.
For the mapping
$\bigwedge:\Fin(M)\to M$,
$$
D_\wedge(X,A)=\bigwedge X\ominus\bigl((\bigwedge X)\wedge(\bigvee A\setminus X)\bigr).
$$
\end{proposition}
\begin{proof}
The proof goes by induction with respect to $|A\setminus X|$.

If $|A\setminus X|=0$, then $A=X$ and
$$
D_\wedge(X,A)=D_\wedge(X,X)=\bigwedge X.
$$
For the right-hand side,
$$
\bigwedge X\ominus\bigl((\bigwedge X)\wedge(\bigvee A\setminus X)\bigr)=
\bigwedge X\ominus\bigl((\bigwedge X)\wedge 0\bigr)=
\bigwedge X.
$$

Let $n\in\mathbb N$, suppose that the Proposition is true for all pairs $(X,A)$ with
$|A\setminus X|\leq n$. Let $X,A_1\in\Fin(M)$ be such that $X\subseteq A_1$ and
$|A_1\setminus X|=n+1$. Pick $c\in A_1\setminus X$ and put $A:=A_1\setminus\{c\}$.
Then $c\notin A$ and $A_1=A\cup\{c\}$. 

By Lemma~\ref{lemma:formal},
$$
D_\wedge(X,A)=D_\wedge(X,A\cup\{c\})+D_\wedge(X\cup\{c\},A\cup\{c\}),
$$
hence
$$
D_\wedge(X,A\cup\{c\})=D_\wedge(X,A)-D_\wedge(X\cup\{c\},A\cup\{c\}).
$$
To abbreviate, let us write
$x:=\bigwedge X$, $a:=\bigvee A\setminus X$.
Note that $\bigvee (A\cup\{c\})\setminus(X\cup\{c\})=a$ and that
$\bigvee (A\cup\{c\})\setminus X=a\vee c$.
We need to prove that
$$
D_\wedge(X,A\cup\{c\})=x\ominus x\wedge(a\vee c).
$$
By the induction hypothesis, we may write
\begin{align*}
D_\wedge(X,A)=&x\ominus x\wedge a\\
D_\wedge(X\cup\{c\},A\cup\{c\})=&(x\wedge c)\ominus (x\wedge c\wedge a),
\end{align*}
therefore
$$
D_\wedge(X,A\cup\{c\})=
	(x\ominus x\wedge a)-\bigl((x\wedge c)\ominus (x\wedge c\wedge a)\bigr).
$$
Thus, it remains to prove that
$$
x\ominus x\wedge(a\vee c)=
	(x\ominus x\wedge a)-\bigl((x\wedge c)\ominus (x\wedge c\wedge a)\bigr),
$$
that means,
$$
\bigl((x\wedge c)\ominus (x\wedge c\wedge a)\bigr)+\bigl(x\ominus x\wedge(a\vee c)\bigr)=x\ominus x\wedge a.
$$
Since $M$ is an MV-effect algebra, we may use the 
equality \ref{eq:MV} and distributivity of $M$ as a lattice to compute
\begin{align*}
(x\wedge c)\ominus (x\wedge c\wedge a)=
\bigl((x\wedge c)\ominus \bigl((x\wedge c)\wedge (x\wedge a)\bigr)\bigr)=
\bigl((x\wedge c)\vee(x\wedge a)\bigr)\ominus(x\wedge a)=\\
=\bigl(x\wedge(c\vee a)\bigr)\ominus (x\wedge a)=
x\wedge(a\vee c)\ominus (x\wedge a),
\end{align*}
hence
\begin{align*}
\bigl((x\wedge c)\ominus (x\wedge c\wedge a)\bigr)+x\ominus x\wedge(a\vee c)=\\
=\bigl(x\wedge(a\vee c)\ominus (x\wedge a)\bigr)+\bigl(x\ominus x\wedge(a\vee c)\bigr)=
x\ominus x\wedge a.
\end{align*}
\end{proof}

\begin{corollary}\label{coro:wedgeiscm}
Let $M$ be an MV-effect algebra. The mapping $\bigwedge:\Fin(M)\to M$ is
a witness mapping.
\end{corollary}
\begin{proof}
Clearly, the conditions (A1) and (A2) are satisfied.

Moreover, for any $(X,A)\in\Fin(M)$, 
$\bigwedge X\geq \bigl((\bigwedge X)\wedge(\bigvee A\setminus X)\bigr)$.
By Proposition~\ref{prop:Dwedge},
$$
D_\wedge(X,A)=\bigwedge X\ominus\bigl((\bigwedge X)\wedge(\bigvee A\setminus X)\bigr).
$$
Therefore, $D_\wedge(X,A)\geq 0$ and we see that (A3) is satisfied.
\end{proof}

\begin{corollary}
\label{coro:mvrange}
Let $E$ be an interval effect algebra,
let $M$ be an MV-effect algebra. Let $\phi:M\to E$ be a morphism of effect algebras.
Every $S\subseteq \phi(M)$ admits a witness mapping.
\end{corollary}
\begin{proof}
By 
Corollary~\ref{coro:wedgeiscm}, $M$ admits a witness mapping.
By Proposition~\ref{prop:pushbeta}, $\phi(M)$ admits a witness mapping.
By Proposition~\ref{prop:restriction}, $S$ admits a witness mapping.
\end{proof}
\begin{corollary}\label{coro:coexwitness}
Every coexistent subset of an interval effect algebra admits a witness mapping.
\end{corollary}
\begin{proof}
Just let the MV-effect algebra $M$ of Corollary~\ref{coro:mvrange} be a Boolean algebra.
\end{proof}
Another natural example of a witness mapping is given by the following proposition,
which can be considered as a generalization of Theorem 2.2 of \cite{Gud:CoQE}.
\begin{proposition}\label{prop:prodiswitness}
Let $S$ be a pairwise commuting subset of $\mathcal E(\mathbb H)$.

Let $\Pi:\Fin(S)\to \mathcal E(\mathbb H)$ be 
given by
$$
\Pi(\{x_1,\dots,x_n\})=x_1.\dots.x_n.
$$

$\Pi$ is a witness mapping.
\end{proposition}
\begin{proof}
The proof goes by induction with respect to
$|A\setminus X|$.

For $A=X$, $D_\Pi(X,A)=D_\Pi(X,X)=\Pi(X)$ and $0\leq\Pi(X)$.

Let $n\in\mathbb N$.
Suppose that, for all $A,X\in\Fin(E)$
such that $|A\setminus X|=n$, $0\leq D_\Pi(X,A)$. 
Let $A_1,X\in\Fin(E)$ be such that 
$|A_1\setminus X|=n+1$. Pick $c\in A_1\setminus X$ and
write $A=A_1\setminus\{c\}$. We see that $c\notin A$ and
that $A_1=A\cup\{c\}$.
We shall prove that $0\leq D_\Pi(X,A\cup\{c\})$.

By Lemma~\ref{lemma:formal},
\begin{equation}\label{eq:reapply}
D_\Pi(X,A\cup\{c\})=D_\Pi(X,A)-D_\Pi(X\cup\{c\},A\cup\{c\})
\end{equation}
Let us compute
\begin{align*}
D_\Pi(X\cup\{c\},A\cup\{c\})=
\sum_{X\cup\{c\}\subseteq Z\subseteq A\cup\{c\}}(-1)^{|X\cup\{c\}|+|Z|}\Pi(Z)=\\
=\sum_{X\subseteq Y\subseteq A}(-1)^{|X\cup\{c\}|+|Y\cup\{c\}|}\Pi(Y\cup\{c\})=
\sum_{X\subseteq Y\subseteq A}(-1)^{|X|+|Y|}\Pi(Y\cup\{c\})=\\
=c.\sum_{X\subseteq Y\subseteq A}(-1)^{|X|+|Y|}\Pi(Y)=c.D_\Pi(X,A).
\end{align*}
Thus, by \ref{eq:reapply} we obtain
$$
D_\Pi(X,A\cup\{c\})=D_\Pi(X,A)-c.D_\Pi(X,A)=(I-c).D_\Pi(X,A)
$$
where $I$ is the identity operator.
By the induction hypothesis, $0\leq D_\Pi(X,A)$,
hence 
$$
0\leq (I-c).D_\Pi(X,A).
$$
\end{proof}

\section{Main result}

The goal of this section is to show that every subset of an interval effect algebra that
admits a witness mapping is coexistent. The main tools we shall use 
to achieve that goal are the following definition and theorem.

\begin{definition}\label{def:projective}(Definition 1.10.36 of \cite{DvuPul:NTiQS}.)
Let $E$ be an effect algebra.
\begin{enumerate}
\item[(a)]
Let $(H,\leq)$ be a directed set and $\Omega_i$ be a finite set for each $i\in H$.
Whenever $i,j\in H$ and $i\leq j$, let there be a mapping $g_{i,j}:\Omega_j\to\Omega_i$ and
denote by $\mathcal G:=\{g_{i,j}:i\leq j\}$ the collection of all such mappings.
The pair $((\Omega_i)_{i\in H},\mathcal G)$ is called a {\em projective system of finite sets} if the following conditions hold:
\begin{enumerate}
\item[(i)]
	$g_{i,i}$ is the identity map on $\Omega_i$,
\item[(ii)]
	$g_{i,j}\circ g_{j,k}=g_{i,k}$, whenever $i\leq j\leq k$.
\end{enumerate}
\item[(b)]
We say that $((\Omega_i,\alpha_i),\mathcal G)$ is a {\em projective system of simple observables}
if $((\Omega_i)_{i\in H},\mathcal G)$ is a projective system of finite sets and for each
$i\in H$, $\alpha_i:2^{\Omega_i}\to E$ is a simple observables such that
the following compatibility condition holds:
\begin{enumerate}
\item[(iii)] For all $X\in 2^{\Omega_i}$, $\alpha_i(X)=\alpha_{j}(g_{i,j}^{-1}(X))$ whenever $i\leq j$. 
\end{enumerate}
\end{enumerate}
\end{definition}

\begin{theorem}(Theorem 1.10.37 of \cite{DvuPul:NTiQS}.)\label{thm:projective}
Let $E$ be an effect algebra. A subset $S$ of $E$ is coexistent if and only if there is
a projective system of simple observables $((\Omega_i,\alpha_i),\mathcal G)$ such that
for every $a\in S$ there is an $i\in H$ such that $a\in\ran(\alpha_i)$.
\end{theorem}

Let us start with showing that every witness mapping $\beta:\Fin(S)\to E$ 
gives rise to a family of simple observables $\alpha_A:2^{2^A}\to E$, for every $A\in\Fin(S)$.

\begin{lemma}
\label{lemma:second}
Let $C,A,X\in\Fin(S)$ be such that $X\subseteq A$ and 
$C\cap A=\emptyset$.
Then $(D_\beta(X\cup Y,A\cup C))_{Y\subseteq C}$ is an orthogonal
family and
$$
\bigoplus_{Y\subseteq C}D_\beta(X\cup Y,A\cup C)=D_\beta(X,A).
$$
\end{lemma}
\begin{proof}
The proof goes by induction with respect to $|C|$.

For $C=\emptyset$, the lemma is trivially true.
Let $C$ be such that $|C|=n$ and let $c\in S,c\not\in A\cup C$.
Let us consider the family
$$
\bigl(D_\beta(X\cup Z,A\cup C\cup\{c\})\bigr)_{Z\subseteq C\cup\{c\}}.
$$
For each $Z\subseteq C\cup\{c\}$, either $c\in Z$ or
$c\not\in Z$, so either $Z=Y\cup\{c\}$ or $Z=Y$, for some
$Y\subseteq C$. Therefore, we can write
\begin{align*}
\bigl(D_\beta(X\cup Z,A\cup C\cup\{c\})\bigr)_{Z\subseteq C\cup\{c\}}=\\
=\bigl(D_\beta(X\cup Y,A\cup C\cup\{c\}),
	D_\beta(X\cup Y\cup\{c\},A\cup C\cup\{c\})\bigr)_{Y\subseteq C}.
\end{align*}
By Proposition~\ref{prop:base},
$$
D_\beta(X\cup Y,A\cup C\cup\{c\})\oplus
	D_\beta(X\cup Y\cup\{c\},A\cup C\cup\{c\})=
D_\beta(X\cup Y,A\cup C).
$$
It only 
remains to apply the induction hypothesis to finish the proof.
\end{proof}
\begin{corollary}
\label{coro:decomposition}
For every $A\in\Fin(S)$, $\bigl(D_\beta(X,A)\bigr)_{X\subseteq A}$ is a
decomposition of unit.
\end{corollary}
\begin{proof}
By Lemma~\ref{lemma:second},
$$
\bigoplus_{X\subseteq A}D_\beta(\emptyset\cup X,\emptyset\cup A)
=D_\beta(\emptyset,\emptyset)=\beta(\emptyset)=1
$$
\end{proof}
\begin{corollary}
\label{coro:betaAobservable}
For every $A\in\Fin(S)$, the mapping
$\alpha_A:2^{(2^A)}\to E$ given by
$$
\alpha_A(\mathbb X)=\bigoplus_{X\in\mathbb X}D_\beta(X,A)
$$
is a simple observable.
\end{corollary}
\begin{proof}
The atoms of $2^{(2^A)}$ are of the form $\{X\}$, where
$X\subseteq A$.
By Corollary~\ref{coro:decomposition},
$(\alpha_A(\{X\}):X\subseteq A)$ is a decomposition of unit;
the remainder of the proof is trivial.
\end{proof}

\begin{theorem}
\label{thm:main}
Let $E$ be an interval effect algebra. $S\subseteq E$ admits a witness
mapping if and only if $S$ is coexistent.
\end{theorem}
\begin{proof}
One implication is Corollary \ref{coro:coexwitness}.

To prove the other implication, we shall apply Theorem \ref{thm:projective}.
Let $\beta$ be a witness mapping.
\begin{itemize}
\item $H=\Fin(S)$, ordered by inclusion.
\item For all $A\in H$, $\Omega_A:=2^A$.
\item For $U,V\in H$ with $U\subseteq V$, $g_{U,V}:\Omega_{V}\to\Omega_{U}$
is given by the rule
$g_{U,V}(X)=X\cap U$ and $\mathcal G$ is the collection of all such $g_{U,V}$.
\item For all $A\in H$, $\alpha_A:2^{\Omega(A)}\to E$ is given by the
rule
$$
\alpha_A(\mathbb X)=\bigoplus_{X\in\mathbb X}D_\beta(X,A)
$$
\end{itemize}

Let us prove that 
$((\Omega_A,\alpha_a),\mathcal G)$ is a projective system of simple observables.
The conditions (i) and (ii) of Definition \ref{def:projective} are easy to check.
By Corollary 7, every $\alpha_A$ is an observable.
Since
$$
\alpha_{\{a\}}(\{\{a\}\})=D_\beta(\{a\},\{a\})=\beta(\{a\})=a,
$$
every $a\in S$ is in the range of the observable 
$\alpha_{\{a\}}$.
It remains to prove the condition (iii).

It is easy to see that
$$
g_{U,V}^{-1}(\mathbb X)=\{X\cup C_0:X\in\mathbb X
	\text{ and }C_0\subseteq (V\setminus U)\}
$$

For all $\mathbb X\in 2^{(2^U)}$,
\begin{align*}
\alpha_V\bigl(g_{U,V}^{-1}(\mathbb X)\bigr)=
	\alpha_V\bigl(\{X\cup C_0:X\in\mathbb X
	\text{ and }C_0\subseteq (V\setminus U)\}\bigr)=\\
=\bigoplus\bigl(
	D_\beta(X\cup C_0,V):X\in\mathbb X
	\text{ and }C_0\subseteq (V\setminus U)
	\bigr)=\\
=\bigoplus_{X\in\mathbb X}\bigl(
\bigoplus_{C_0\subseteq (V\setminus U)}
D_\beta(X\cup C_0,B)\bigr)
\end{align*}
Put $Y:=C_0$, $C:=V\setminus U$; by Lemma~\ref{lemma:second},
$$
\bigoplus_{C_0\subseteq (V\setminus U)}
D_\beta(X\cup C_0,B)=D_\beta(X,U).
$$
Therefore,
$$
\alpha_V\bigl(g_{U,V}^{-1}(\mathbb X)\bigr)=\bigoplus_{X\in\mathbb X} D_\beta(X,U)=
\alpha_U(\mathbb X),
$$
and the condition (iii) is satisfied.

All the conditions of Theorem \ref{thm:projective} are satisfied. Therefore,
$S$ is a coexistent subset of $E$.
\end{proof}
\begin{corollary}\label{coro:mviscoex}
Every MV-effect algebra is a coexistent subset of itself.
\end{corollary}
\begin{proof}
By Corollary \ref{coro:wedgeiscm} and Theorem \ref{thm:main}.
\end{proof}
We note that Corollary \ref{coro:mviscoex} was already proved in \cite{Jen:BARGbMVEA},
using a different method.
\begin{corollary}
Let $\mathbb H$ be a Hilbert space.
Every pairwise commuting subset of $\mathcal E(\mathbb H)$ is coexistent.
\end{corollary}
\begin{proof}
By Corollary \ref{prop:prodiswitness} and Theorem \ref{thm:main}.
\end{proof}

\end{document}